\documentclass[11pt,a4paper]{article}
\pdfoutput=1
\usepackage[left=1in,right=1in,top=1in,bottom=1in]{geometry}
\usepackage[T1]{fontenc}
\usepackage{graphicx}
\usepackage{float}
\usepackage[tbtags]{amsmath}
\usepackage{amssymb}
\usepackage{amsthm}
\usepackage{hyperref}
\usepackage[numbers]{natbib}
\newtheorem{theorem}{Theorem}
\newtheorem{corollary}{Corollary}
\newtheorem{lemma}{Lemma}
\theoremstyle{definition}
\newtheorem{definition}{Definition}
\newtheorem{assumption}{Assumption}
\newtheorem{example}{Example}
\def \R{\mathbb R}
\newcommand{\sset}[1]{\left\{ #1\right\}}
\newcommand{\fwh}[1]{\; \left| \; #1 \right.}

\newcommand{\rev}{\ensuremath{\text{\rm\sc Rev}}}
\newcommand{\union}{\cup} 
\newcommand{\map}{\longrightarrow}
\newcommand{\inters}{\cap}    
\newcommand{\vecc}[1]{\ensuremath{\mathbf{#1}}}
\newcommand{\slice}[3]{\left.#1\right|_{{#2}:#3}}  

\title{Selling Two Goods Optimally\footnote{The research leading to these results has received funding from the European Research Council under the European Union's Seventh Framework Programme (FP7/2007-2013)/ERC grant agreement no.\ 321171.
\newline \indent
A preliminary version of this paper appeared in ICALP 2015~\citep{gk2015_icalp}.}}
\author{Yiannis Giannakopoulos\thanks{Department of Computer Science, University of Oxford. Email: \href{mailto:ygiannak@cs.ox.ac.uk}{\nolinkurl{ygiannak@cs.ox.ac.uk} }} \and Elias Koutsoupias\thanks{Department of Computer Science, University of Oxford. Email: \href{mailto:elias@cs.ox.ac.uk}{\nolinkurl{elias@cs.ox.ac.uk} }}}
\date{November 21, 2015}
\begin{document}
\maketitle

\begin{abstract}
We provide sufficient conditions for revenue maximization in a two-good monopoly where the buyer's values for the items come from independent (but not necessarily identical) distributions over bounded intervals. Under certain distributional assumptions, we give exact, closed-form formulas for the prices and allocation rule of the optimal selling mechanism. As a side result we give the first example of an optimal mechanism in an i.i.d.\ setting over a support of the form $[0,b]$ which is \emph{not} deterministic. Since our framework is based on duality techniques, we were also able to demonstrate how slightly relaxed versions of it can still be used to design mechanisms that have very good approximation ratios with respect to the optimal revenue, through a ``convexification'' process. 
\end{abstract}

\section{Introduction}

The problem of designing auctions that maximize the seller's revenue in settings with many heterogeneous goods has attracted a large amount of interest in the last years, both from the Computer Science as well as the Economics community (see e.g.~\citep{Manelli:2006vn,Pavlov:2011fk,Hart:2012uq,Hart:2012zr,Daskalakis:2012fk,Daskalakis:2013vn,gk2014,Menicucci:2014jl,Daskalakis:2014fk}). Here the seller faces a buyer whose true values for the $m$ items come from a probability distribution over $\R_+^m$ and, based only on this incomplete prior knowledge, he wishes to design a selling mechanism that will maximize his expected revenue. For the purposes of this paper, the prior distribution is a product one, meaning that the item values are independent. The buyer is additive, in the sense that her happiness from receiving any subset of items is the sum of her values of the individual items in that bundle. The buyer is also selfish and completely rational, thus willing to lie about her true values if this is to improve her own happiness. So, the seller should also make sure to give the right incentives to the buyer in order to avoid manipulation of the protocol by misreporting. 

The special case of a single item has been very well understood since the seminal work of~\citet{Myerson:1981aa}. However, when one moves to settings with multiple goods, the problem becomes notoriously difficult and novel approaches are necessary. Despite the significant effort of the researchers in the field, essentially only specialized, partial results are known: there are exact solutions for two items in the case of identical uniform distributions over unit-length intervals~\citep{Pavlov:2011fk,Manelli:2006vn}, exponential over $[0,\infty)$ \citep{Daskalakis:2013vn} or identical Pareto distributions with tail index parameters $\alpha\geq 1/2$ \citep{Hart:2012uq}. For more than two items, optimal results are only known for uniform values over the unit interval~\citep{gk2014}, and due to the difficulty of exact solutions most of the work focuses in showing approximation guarantees for simple selling mechanisms~\citep{Hart:2012uq,Li:2013ty,Babaioff:2014ys,g2014,Bateni:2014ph,Rubinstein:2015kx}. This difficulty is further supported by the complexity ($\#P$-hardness) results of~\citet{Daskalakis:2012fk}. It is important to point out that even for two items \emph{we know of no general and simple, closed-form conditions framework under which optimality can be extracted when given as input the item distributions, in the case when these are not necessarily identical.} This is our goal in the current paper. 
 
\paragraph{Our contribution}
We introduce general but simple and clear, closed-form distributional conditions that can guarantee optimality and immediately give the form of the revenue-maximizing selling mechanism (its payment and allocation rules), for the setting of two goods with values distributed over bounded intervals (Theorem~\ref{th:characterization_main}). For simplicity and a clearer exposition we study distributions supported over the real unit interval $[0,1]$. By scaling, the results generalize immediately to intervals that start at $0$, but more work would be needed to generalize them to arbitrary intervals. We use the closed forms to get optimal solutions for a wide class of distributions satisfying certain simple analytic assumptions (Theorem~\ref{th:characterization_iid} and Sect.~\ref{sec:non-iid}). As useful examples, we provide exact solutions for families of monomial ($\propto x^c$) and exponential ($\propto e^{-\lambda x}$) distributions (Corollaries~\ref{th:optimal_two_power} and \ref{th:optimal_two_expo} and Sect.~\ref{sec:non-iid}), and also near-optimal results for power-law ($\propto (x+1)^{-\alpha}$) distributions (Sect.~\ref{sec:approximate_convex_fail}). This last approximation is an application of a more general result (Theorem~\ref{th:two_iid_approx}) involving the relaxation of some of the conditions for optimality in the main Theorem~\ref{th:characterization_main}; the ``solution'' one gets in this new setting might not always correspond to a feasible selling mechanism, however it still provides an upper bound on the optimal revenue as well as hints as to how to design a well-performing mechanism, by ``convexifying'' it into a feasible mechanism (Sect.~\ref{sec:approximate_convex_fail}).

Particularly for the family of monomial distributions it turns out that the optimal mechanism is a very simple deterministic mechanism that offers to the seller a menu of size just $4$ (using the menu-complexity notion of Hart and Nisan \citep{Hart:2012ys,Wang:2013ab}): fixed prices for each one of the two items and for their bundle, as well as the option of not buying any of them. For other distributions studied in the current paper randomization is essential for optimality, as is generally expected in such problems of multidimensional revenue maximization (see e.g.~\citep{Hart:2012zr,Pavlov:2011fk,Daskalakis:2013vn}). For example, this is the case for two i.i.d. exponential distributions over the unit interval $[0,1]$, which gives the first such example where determinism is suboptimal even for regularly\footnote{A probability distribution $F$ is called \emph{regular} if $t-\frac{1-F(t)}{f(t)}$ is increasing. This quantity is known as the \emph{virtual valuation}.} i.i.d.\ items. 
A point worth noting here is the striking difference between this result and previous results~\citep{Daskalakis:2013vn,g2014} about i.i.d.\ exponential distributions which have as support the entire $\R_+$: the optimal selling mechanism there is the deterministic one that just offers the full bundle of both items.

Although the conditions that the probability distributions must satisfy are quite general, they leave out a large class of distributions. For example, they do not apply to power-law distributions with parameter $\alpha>2$. In other words, this work goes some way towards the complete solution for arbitrary distributions for two items, but the general problem is still open. In this paper, we opted towards simple conditions rather than full generality, but we believe that extensions of our method can generalize significantly the range of distributions; we expect that a proper ``ironing'' procedure will enable our technique to resolve the general problem for two items.

\paragraph{Techniques}
The main result of the paper (Theorem~\ref{th:characterization_main}) is proven by utilizing the \emph{duality} framework of~\citep{gk2014} for revenue maximization, and in particular using complementarity: the optimality of the proposed selling mechanism is shown by verifying the existence of a dual solution with which they satisfy together the required complementary slackness conditions of the duality formulation. Constructing these dual solutions explicitly seems to be a very challenging task and in fact there might not even be a concise way to do it, especially in closed-form. So instead we just prove the existence of such a dual solution, using a \emph{max-flow min-cut} argument as main tool (Lemma~\ref{lemma:coloring}, Fig.~\ref{fig:flows_graph}). This is, in a way, an abstraction of a technique followed in~\citep{gk2014} for the case of uniform distributions which was based on Hall's theorem for bipartite matchings. Since here we are dealing with general and non-identical distributions, this kind of refinement is essential and non-trivial, and in fact forms the most technical part of the paper. Our approach has a strong geometric flavor, enabled by introducing the notion of the \emph{deficiency} of a two-dimensional body (Definition~\ref{def:deficiency}, Lemma~\ref{lemma:no_positive_def}), which is inspired by classic matching theory~\citep{Ore:1955fk,Lovasz:1986qf}. 

\subsection{Model and Notation}
We study a two-good monopoly setting in which a seller deals with a buyer who has values $x_1, x_2\in I$ for the items, where $I=[0,1]$. The seller has only an incomplete knowledge of the buyer's preference, in the form of two independent distributions (with densities) $f_1$, $f_2$ over $I$ from which $x_1$ and $x_2$ are drawn, respectively. The cdf of $f_j$ will be denoted by $F_j$. As in the seminal work of~\citet{Myerson:1981aa}, the density functions will be assumed to be absolutely continuous and positive.
We will also use vector notation $\vecc x=(x_1,x_2)$. For any item $j\in\{1,2\}$, index $-j$ will refer the complementary item, that is $3-j$, and as it's standard in game theory $\vecc x_{-j}=x_{-j}$ will denote the remaining of vector $\vecc x$ if the $j$-th coordinate is removed, so $\vecc x=(x_j,x_{-j})$ for any $j=1,2$.

The seller's goal is to design a selling mechanism that will maximize his revenue. Without loss\footnote{This is due to the celebrated Revelation Principle~\citep{Myerson:1981aa}.} we can focus on direct-revelation mechanisms: the bidder will be asked to submit bids $b_1,b_2$ and the mechanism consists simply of an allocation rule $a_1,a_2:I^2\to I$ and a payment function $p:I^2\to\R_+$ such that $a_j(b_1,b_2)$ is the probability of item $j$ being sold to the buyer (notice how we allow for randomized mechanisms, i.e.~lotteries) and $p(b_1,b_2)$ is the payment that the buyer expects to pay; it is easier to consider the expected payment for all allocations, rather than individual payments that depend on the allocation of items. The reason why the bids $b_j$ are denoted differently than the original values $x_j$ for the items is that, since the bidder is a rational and selfish agent, she might lie and misreport $b_j\neq x_j$ if this is to increase her personal gain given by the quasi-linear \emph{utility} function 
\begin{equation}
\label{eq:utility}
u(\vecc b;\vecc x)\equiv a_1(\vecc b) x_1+a_2(\vecc b) x_2-p(\vecc b),
\end{equation}
the expected happiness she'll receive by the mechanism minus her payment. 
Thus, we will demand our selling mechanisms to satisfy the following standard properties: 
\begin{itemize}
\item \emph{Incentive Compatibility (IC)}, also known as truthfulness, saying that the player would have no incentive to misreport and manipulate the mechanism, i.e.\ her utility is maximized by truth-telling: $u(\vecc b;\vecc x)\leq u(\vecc x;\vecc x)$ 
\item \emph{Individual Rationality (IR)}, saying that the buyer cannot harm herself just by truthfully participating in the mechanism: $u(\vecc x;\vecc x)\geq 0$.
\end{itemize}
It turns out the critical IC property comes without loss\footnote{Also due to the Revelation Principle.} for our revenue-maximization objective, so for now on we will only consider truthful mechanisms, meaning we can also relax the notation $u(\vecc b;\vecc x)$ to just $u(\vecc x)$.

There is a very elegant and helpful analytic characterization of truthfulness, going back to~\citet{Rochet:1985aa} (for a proof see e.g.~\citep{Hart:2012uq}), which states that the player's utility function must be \emph{convex} and that the allocation probabilities are simply given by the utility's derivatives, i.e.\ $\partial u(\vecc x)/\partial x_j=a_j(\vecc x)$. Taking this into consideration and rearranging~\eqref{eq:utility} with respect to the payment, we define
$$
\mathcal R_{f_1,f_2}(u)\equiv \int_0^1\int_0^1\left(\frac{\partial u(\vecc x)}{\partial x_1}x_1+\frac{\partial u(\vecc x)}{\partial x_2}x_2-u(\vecc x) \right)f_1(x_1)f_2(x_2)\,dx_1\,dx_2
$$
for every absolutely continuous function $u:I^2\map\R_+$. If $u$ is convex with partial derivatives in $[0,1]$ then $u$ is a valid utility function and \emph{$\mathcal R_{f_1,f_2}(u)$ is the expected revenue of the seller under the mechanism induced by $u$}. Let $\rev(f_1,f_2)$ denote the best possible such revenue, i.e.\ the supremum of $\mathcal R_{f_1,f_2}(u)$ when $u$ ranges over the space of all feasible utility functions over $I^2$. So the problem we want to deal with in this paper is exactly that of $\sup_u \mathcal R_{f_1,f_2}(u)$.

We now present the condition on the probability distributions which will enable our technique to provide a closed-form of the optimal auction.

\begin{assumption}
\label{assume:upwards_def}
\label{assume:regularity}
The probability distributions $f_1,f_2$ are such that functions $h_{f_1,f_2}(\vecc x)-f_2(1)f_1(x_1)$ and $h_{f_1,f_2}(\vecc x)-f_1(1)f_2(x_2)$ are nonnegative, where
\begin{equation}
\label{eq:def_h}
h_{f_1,f_2}(\vecc x)\equiv 3 f_1(x_1)f_2(x_2)+x_1f_1'(x_1)f_2(x_2)+x_2f'_2(x_2)f_1(x_1).
\end{equation}
Function $h_{f_1,f_2}$ will also be assumed to be absolutely continuous with respect to each of its coordinates.
\end{assumption}
We will drop the subscript $f_1,f_2$ in the above notations whenever it is clear which distributions we are referring to. Assumption~\ref{assume:upwards_def} is a slightly stronger condition than $h(\vecc x)\geq 0$ which is a common regularity assumption in the economics literature for multidimensional auctions with $m$ items: $(m+1)f(\vecc x)+\nabla f(\vecc x)\cdot \vecc x\geq 0$, where $f$ is the joint distribution for the item values (see e.g.~\citep{Manelli:2006vn,Pavlov:2011fk,McAfee:1988nx}). In fact, \citet{Manelli:2006vn} make the even stronger assumption that for each item $j$, $x_j f_j(x_j)$ is an increasing function. Even more recently, that assumption has also been deployed by \citet{Wang:2013ab} in a two-item setting as one of their sufficient conditions for the existence of optimal auctions with small-sized menus.  It has a strong connection with the standard single-dimensional regularity condition of~\citet{Myerson:1981aa}, since for $m=1$ condition $h(\vecc x)\geq 0$ gives that $f(x)\left( x-\frac{1-F(x)}{f(x)} \right)$ is increasing, thus ensures the single-crossing property of the virtual valuation function (see also the discussion in \citep[Sect.~2]{Manelli:2006vn}). 

Strengthening the regularity condition $h(\vecc x)\geq 0$ to that of Assumption~\ref{assume:upwards_def} is essentially only used 
as a technical tool within the proof of Lemma~\ref{lemma:no_positive_def}, and
as a matter of fact we don't really need it to hold in the entire unit box $I^2$ but just in a critical sub-region $D_{1,2}$ which corresponds to the valuation subspace where both items are sold with probability $1$ (see Fig.~\ref{fig:Exp_Uniform} and Sect.~\ref{sec:partition_optimal}). As mentioned earlier in the Introduction, we introduce this technical conditions in order to simplify our exposition and enforce the clarity of the techniques, but we believe that a proper ``ironing''~\citep{Myerson:1981aa} process can probably bypass these restrictions and generalize our results.
The critical Assumption~\ref{assume:upwards_def} is of course satisfied by all distributions considered in the results of this paper, namely monomial $\propto x^c$ for any power $c\geq 0$ (Corollary~\ref{th:optimal_two_power}), exponential $\propto e^{-\lambda x}$ with rates $\lambda\leq 1$ (Corollary~\ref{th:optimal_two_expo}), power-law $\propto (t+1)^{-\alpha}$ with parameters $\alpha\leq 2$ (Example~\ref{example:power-law}), as well as combinations of these (see Example~\ref{example:uniform-expo}). However, there is still a large class of distributions not captured by Assumption~\ref{assume:upwards_def} as it is, e.g.\ exponential with rates larger than $1$, power-law with parameters greater than $2$ and some beta-distributions (take, for example, $\propto x^2(1-x)^2$). See Footnote~\ref{foot:alter-assumption-monotone} for an alternative condition that can replace Assumption~\ref{assume:upwards_def}.
\section{Sufficient Conditions for Optimality}
This section is dedicated to proving the main result of the paper:
\begin{theorem}
\label{th:characterization_main}
If there exist decreasing, concave functions $s_1,s_2:I\to I$, with $s_1'(t),s_2'(t)> -1$ for all $t\in I$, such that for almost every\footnote{Everywhere except a subset of zero Lebesgue measure.} (a.e.) $x_1,x_2\in I$
\begin{equation}
\label{eq:1slice_gen_functions}
\frac{s_1(x_2)f_1(s_1(x_2))}{1-F_1(s_1(x_2))} =2+\frac{x_2f_2'(x_2)}{f_2(x_2)}
\quad\text{and}\quad
\frac{s_2(x_1)f_2(s_2(x_1))}{1-F_2(s_2(x_1))} =2+\frac{x_1f_1'(x_1)}{f_1(x_1)}, 
\end{equation}
then 
there exists a constant $p\in[0,2]$ such that 
\begin{equation}
\label{eq:2slice_gen}
\int_{D}h(\vecc x)\,dx_1\,dx_2 
=f_1(1)+f_2(1)
\end{equation}
where $D$ is the region of $I^2$ enclosed by curves\footnote{See Fig.~\ref{fig:Exp_Uniform}.} $x_1+x_2=p$, $x_1=s_1(x_2)$ and $x_2=s_2(x_1)$ and including point $(1,1)$, i.e.~$D=\sset{\vecc x\in I\fwh{x_1+x_2\geq p\lor x_1\geq s_1(x_2) \lor x_2\geq s_2(x_1)}}$,   
and the optimal selling mechanism is given by the utility function
\begin{equation}
\label{eq:optimal_auction_gen}
u(\vecc x)=\max\sset{0,x_1-s_1(x_2),x_2-s_2(x_1),x_1+x_2-p}.
\end{equation}
In particular, if $p\leq\min\sset{s_1(0),s_2(0)}$, then the optimal mechanism is the deterministic full-bundling with price $p$. 
\end{theorem}

Notice that for any $s\in I$ we have
\begin{align*}
\int_s^1h(\vecc x)\,dx_1 &=\int_s^1 3f_1(x_1)f_2(x_2)+x_1f_1'(x_1)f_2(x_2)+x_2f_2'(x_2)f_1(x_1) \,dx_1\\
				&=3f_2(x_2)(1-F_1(s))+f_2(x_2)\int_s^1x_1f_1'(x_1)\,dx_1+x_2f_2'(x_2)(1-F_1(s))\\
				&=3f_2(x_2)(1-F_1(s))+f_2(x_2)\left(\left[x_1f_1(x_1)\right]_s^1-(1-F_1(s))\right)+x_2f_2'(x_2)(1-F_1(s))\\
				&=2f_2(x_2)(1-F_1(s))+f_2(x_2)(f_1(1)-sf_1(s))+x_2f_2'(x_2)(1-F_1(s))\\
				&=(1-F_1(s))f_2(x_2)\left[2+\frac{x_2f_2'(x_2)}{f_2(x_2)}-\frac{sf_1(s)}{1-F_1(s)} \right] +f_1(1)f_2(x_2)
\end{align*}
which means that an equivalent way of looking at~\eqref{eq:1slice_gen_functions} is, more simply, by
\begin{equation}
\label{eq:1slice_gen_integrals}
\int_{s_1(x_2)}^1h(\vecc x)\,dx_1=f_1(1)f_2(x_2)
\quad\text{and}\quad
\int_{s_2(x_1)}^1h(\vecc x)\,dx_2=f_2(1)f_1(x_1).
\end{equation}
This also means that~\eqref{eq:1slice_gen_integrals} can take the place of~\eqref{eq:1slice_gen_functions} in the statement of Theorem~\ref{th:characterization_main} whenever this gives an easier way to solve for functions $s_1$ and $s_2$.
\subsection{Partitioning of the Valuation Space}
\label{sec:partition_optimal}
\begin{figure}
\centering
\includegraphics[width=10cm]{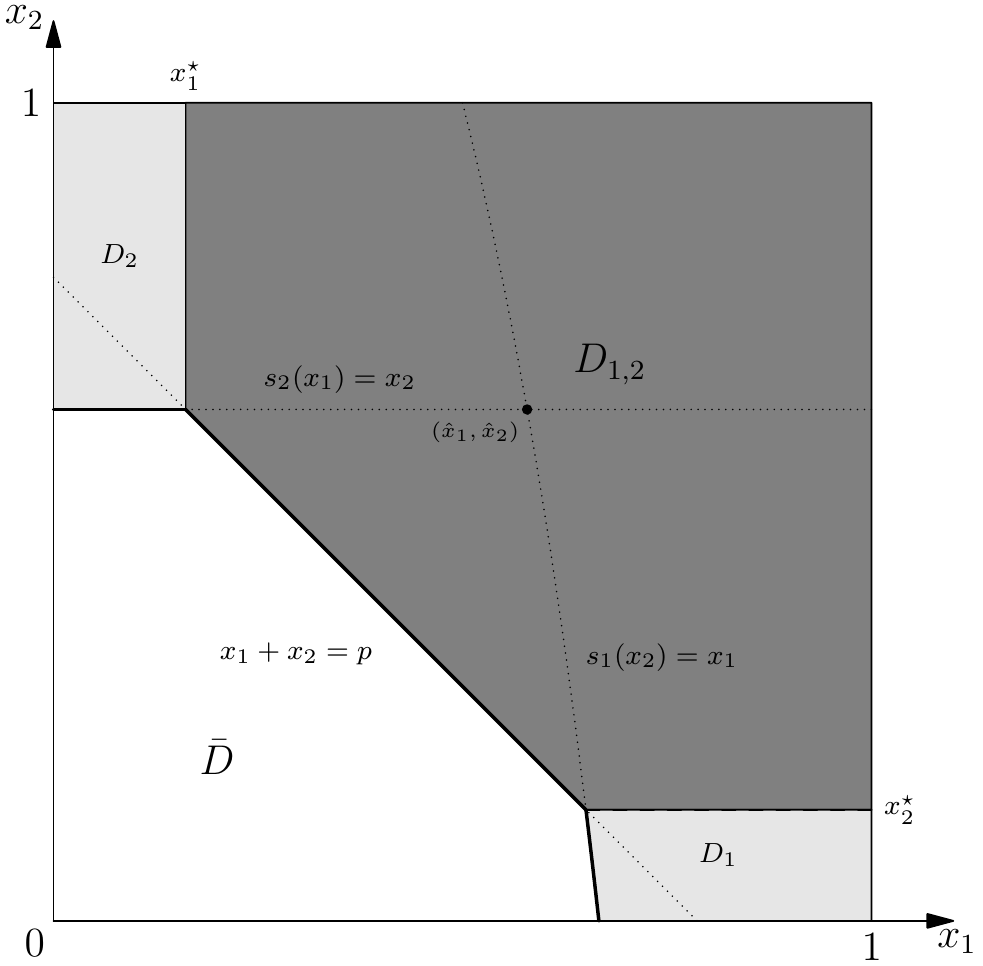}
\caption{\footnotesize The valuation space partitioning of the optimal selling mechanism for two independent items, one following a uniform distribution and the other an exponential with parameter $\lambda=1$. Here $s_1(t)=(2-t)/(3-t)$, $s_2(t)=2-W(2e)\approx 0.625$  and $p\approx 0.787$. In region $D_{1}$ (light grey) item $1$ is sold deterministically and item $2$ with a probability of $-s_1'(x_2)$, in $D_{2}$ (light grey) only item $2$ is sold and region $D_{1,2}$ (dark grey) is where the full bundle is sold deterministically, for a price of $p$.}
\label{fig:Exp_Uniform}
\end{figure}
Due to the fact that the derivatives of functions $s_j$ in Theorem~\ref{th:characterization_main} are above $-1$, each curve $x_1=s_1(x_2)$ and $x_2=s_2(x_1)$ can intersect the full-bundle line $x_1+x_2=p$ at most at a single point. So let $x_2^*=x_2^*(p), x_1^*=x_1^*(p)$ be the coordinates of these intersections, respectively, i.e.~$s_1(x_2^*)=p-x_2^*$ and $s_2(x_1^*)=p-x_1^*$. If such an intersection does not exist, just define $x_2^*=0$ or $x_1^*=0$.

The construction and the optimal mechanism given in Theorem~\ref{th:characterization_main} then gives rise to the following partitioning of the valuation space $I^2$ (see Fig.~\ref{fig:Exp_Uniform}):
\begin{itemize}
\item Region $\bar D=I^2\setminus D$ where no item is allocated
\item Region $D_1=\sset{\vecc x\in I^2\fwh{x_1\geq s_1(x_2)\land x_2\leq x_2^*}}$ where item $1$ is sold with probability $1$ and item $2$ with probability $-s_1'(x_2)$ for a price of $s_1(x_2)-x_2s_1'(x_2)$
\item Region $D_2=\sset{\vecc x\in I^2\fwh{x_2\geq s_2(x_1)\land x_1\leq x_1^*}}$ where item $2$ is sold with probability $1$ and item $1$ with probability $-s_2'(x_1)$ for a price of $s_2(x_1)-x_1s_2'(x_1)$
\item Region $D_{1,2}=D\setminus{D_1\union D_2}=\sset{\vecc x\in I^2\fwh{x_1+x_2\geq p \land x_1\geq x_1^* \land x_2\geq x_2^*}}$ where both items are sold deterministically in a full bundle of price $p$.
\end{itemize}

Under this decomposition, by \eqref{eq:1slice_gen_integrals}:
$$
\int_{D_{1}}h(\vecc x)\,dx_1\,dx_2=\int_{0}^{x_2^*}\int_{s_1(x_2)}^1h(\vecc x)\,dx_1\,dx_2=f_1(1)F_2(x_2^*)
$$
so expression~\eqref{eq:2slice_gen} can be written equivalently as
\begin{equation}
\label{eq:2slice_gen_bundle_region}
\int_{D_{1,2}}h(\vecc x)\,dx_1\,dx_2 
= f_1(1)(1-F_2(x_2^*))+f_2(1)(1-F_1(x_1^*)).
\end{equation}
\subsection{Duality}
\label{sec:duality}
The major underlying tool to prove Theorem~\ref{th:characterization_main} will be the duality framework of~\citep{gk2014}. For completeness we briefly present here the formulation and key aspects, and the interested reader is referred to the original text for further details. 

Remember that the revenue optimization problem we want to solve here is to maximize $\mathcal R(u)$ over the space of all convex functions $u:I^2\map\R_+$ with
\begin{equation}
\label{eq:allocs_probs_01}
0\leq \frac{\partial u(\vecc x)}{\partial x_j}\leq 1,\qquad j=1,2,
\end{equation}
for a.e. $\vecc x\in I^2$. First we relax this problem by dropping the convexity assumption and replacing it with (absolute) continuity. We also drop the lower bound in~\eqref{eq:allocs_probs_01}. Then this new relaxed program is dual to the following: minimize $\int_0^1\int_0^1 z_1(\vecc x)+z_2(\vecc x)\,d\vecc x$ where the new dual variables $z_1,z_2:I^2\map\R_+$ are such that $z_j$ is (absolutely) continuous with respect to its $j$-coordinate and the following conditions are satisfied for all $x_1,x_2\in I$:
\begin{align}
z_j(0,x_{-j}) &=0, &&j=1,2, \label{eq:dual_const_1}\\
z_j(1,x_{-j}) &\geq f_j(1)f_{-j}(x_{-j}), &&j=1,2, \label{eq:dual_const_2}\\
\frac{\partial z_1(\vecc x)}{\partial x_2}+\frac{\partial z_2(\vecc x)}{\partial x_2} &\leq 3 f_1(x_1)f_2(x_2)+ x_1f_1'(x_1)f_2(x_2)+x_2f_1(x_1)f_2'(x_2).\label{eq:dual_const_3}
\end{align}
We will refer to the first optimization problem, where $u$ ranges over the relaxed space of continuous, nonnegative functions with derivatives at most $1$, as the \emph{primal program} and to the second as the \emph{dual}. Intuitively, every dual solution $z_j$ must start at zero and grow all the way up to $f_j(1)f_{-j}(x_{-j})$ while travelling in interval $I$, in a way that the sum of the rate of growth of both $z_1$ and $z_2$ is never faster than the right hand side of~\eqref{eq:dual_const_3}.
In~\citep{gk2014} is proven that indeed these two programs satisfy both weak duality, i.e.~for any feasible $u,z_1,z_2$ we have
$$
\mathcal R(u)\leq \int_{0}^1\int_0^1 z_1(\vecc x)+z_2(\vecc x)\,d\vecc x
$$ 
as well as complementary slackness, in the form of the even stronger following form of $\varepsilon$-complementarity: 

\begin{lemma}[Complementarity]\label{lemma:complementarity}
If $u,z_1,z_2$ are feasible primal and dual solutions, respectively, $\varepsilon>0$ and the following complementarity constraints hold for a.e. $\vecc x\in I^2$,
\begin{align}
  u(\vecc x)   \left( h(\vecc x)
    -\frac{\partial z_1(\vecc x)}{\partial x_1}-\frac{\partial z_2(\vecc x)}{\partial x_2}
  \right) &\leq \varepsilon f_1(x_1)f_2(x_2), \label{eq:e_compl_2}\\
u(1,x_{-j})  \left( z_j(1, x_{-j})
    -f_j(1)f_{-j}(x_{-j}) \right) &\leq \varepsilon f_j(1)f_{-j}(x_{-j}) \label{eq:e_compl_3}, &&j=1,2,\\
  z_j(\vecc x)  \left( 1 - \frac{\partial
        u(\vecc x)}{\partial x_j} \right)   &\leq \varepsilon f_1(x_1)f_2(x_2),  &&j=1,2, \label{eq:e_compl_4}
\end{align}
where $h$ is defined in~\eqref{eq:def_h}, then the values of the primal and dual programs differ by at most
$7\varepsilon$. In particular, if the conditions are satisfied
with $\varepsilon=0$, both solutions are optimal.
\end{lemma}   

Our approach into proving Theorem~\ref{th:characterization_main} will be to show the existence of a pair of dual solutions $z_1,z_2$ with respect to which the utility function $u$ given by the theorem indeed satisfies complementarity. Notice here the existential character of our technique: our duality approach offers the advantage to use the proof of just the existence of such duals, without having to explicitly describe them and compute their objective value in order to prove optimality, i.e.~that the primal and dual objectives are indeed equal. Also notice that the utility function $u$ given by Theorem~\ref{th:characterization_main} is convex by construction, so in case someone shows optimality for $u$ in the relaxed setting, then $u$ must also be optimal among all feasible mechanisms.

Define function $W:I^2\to\R_+$ by
$$
W(\vecc x)=
\begin{cases}
h(\vecc x), &\text{if}\;\; \vecc x\in D,\\
0, &\text{otherwise},
\end{cases} 
$$
where $D$ is defined in Sect.~\ref{sec:partition_optimal} (see Fig.~\ref{fig:Exp_Uniform}).
If one could decompose $W$ into functions $w_1,w_2:I^2\to\R_+$ such that
\begin{align}
w_1(\vecc x)+w_2(\vecc x) &=W(\vecc x)\label{eq:Wdecomp_sum} \\
\int_0^1w_j(\vecc x)\,d x_j &= f_j(1)f_{-j}(x_{-j}) \label{eq:Wdecomp1}, \qquad j=1,2,
\end{align}
for all $\vecc x\in I$, and $w_j$ is almost everywhere continuous with respect  to its $j$-th coordinate, then by defining 
$$
z_j(\vecc x)=\int_0^{x_j} w_j(t,x_{-j})\,dt
$$
we'll have
\begin{align}
\frac{\partial z_1(\vecc x)}{\partial x_1}+\frac{\partial z_2(\vecc x_2)}{\partial x_2} &=
\begin{cases}
 h(\vecc x), & \text{for}\;\; \vecc x\in D,\\
 0, &\text{otherwise},
 \end{cases}
 \label{eq:prop_dual_1}
 \\
z_j(0,x_{-j}) &=0, && j=1,2, \label{eq:prop_dual_2} \\
z_j(1,x_{-j}) &= f_j(1)f_{-j}(x_{-j}), && j=1,2. \label{eq:prop_dual_3} 
\end{align}
If the requirements of Theorem~\ref{th:characterization_main} hold, then it is fairly straightforward to get such a decomposition in certain regions. In particular, we can set $w_1=w_2=0$ in $I^2\setminus D$, $w_1=W=h$ and $w_2=0$ in $D_1$ and $w_2=W=h$ and $w_1=0$ in $D_2$. Then, by~\eqref{eq:1slice_gen_integrals}, it is not difficult to see that indeed conditions~\eqref{eq:Wdecomp_sum}--\eqref{eq:Wdecomp1} are satisfied. However, \emph{it is highly non-trivial how to create such a decomposition in the remaining region $D_{1,2}$} and that is what the proof of Lemma~\ref{lemma:coloring} achieves, with the assistance of the geometric Lemma~\ref{lemma:no_positive_def}, in the remaining of this section. This is the most technical part of the paper.

In any case, if we are able to get such a decomposition, by the previous discussion that would mean that functions $z_1,z_2:I^2\to\R_+$ are \emph{feasible dual} solutions: it is trivial to verify that properties~\eqref{eq:prop_dual_1}--\eqref{eq:prop_dual_3} satisfy the dual constraints \eqref{eq:dual_const_1}--\eqref{eq:dual_const_3}.  But most importantly, the \emph{equalities} in properties~\eqref{eq:prop_dual_1}--\eqref{eq:prop_dual_3} and the way $w_1$ and $w_2$ are defined  in regions $D_1$ and $D_2$ tell us something more: that this pair of solutions would satisfy complementarity with respect to the primal given in~\eqref{eq:optimal_auction_gen} and whose allocation is analyzed in detail in Sect.~\ref{sec:partition_optimal}, thus proving that this mechanism is optimal and thus establishing Theorem~\ref{th:characterization_main}.

\subsection{Deficiency}
\label{sec:nodef}
The following notion will be the tool that gives a very useful geometric interpretation to the rest of the proof of Theorem~\ref{th:characterization_main} and it will be critical into proving Lemma~\ref{lemma:coloring}.   
\begin{definition}
\label{def:deficiency}
For any body $S\subseteq I^2$ define its \emph{deficiency} (with respect to distributions $f_1,f_2$) to be
$$
\delta(S)\equiv \int_S h(\vecc x)\,d\vecc x - f_2(1)\int_{S_1}f_1(x_1)\,dx_1-f_1(1)\int_{S_2}f_2(x_2)\,dx_2,
$$
where $S_1$, $S_2$ denote $S$'s projections to the $x_1$ and $x_2$ axis, respectively.
\end{definition}
\begin{lemma}
\label{lemma:no_positive_def}
If the requirements of Theorem~\ref{th:characterization_main} hold, then no body $S\subseteq D_{1,2}$ has positive deficiency.
\end{lemma}
\begin{proof}
To get to a contradiction, assume that there is body $S\subseteq D_{1,2}$ with $\delta(S)>0$. 
First, we'll show that without loss $S$ can be assumed to be upwards closed. Intuitively, we'll show that one can push mass of $S$ to the right or upwards, without reducing its deficiency. By Assumption~\ref{assume:upwards_def} function $h(\vecc x)-f_2(1)f_1(x_1)$ is nonnegative. Then, if there exists a nonempty horizontal line segment $\slice{S}{x_2}{t}$ of $S$ at some height $x_2=t$, then we can assume that this line segment fills the entire available horizontal space of $D_{1,2}$: if that was not the case, and there existed a small interval $[\alpha,\beta]\times{t}$ that was not in $S$, then we could add it to it, not increasing the projection towards the $x_2$-axis (it is already covered by the other existing points at $x_2=t$) and the projection towards the $x_1$-axis is increased at most by $\beta-\alpha$, leading to a change to the overall deficiency by at most $\int_{\alpha}^{\beta}h(\vecc x)\,dx_1-f_2(1)\int_{\alpha}^{\beta}f_1(x_1)\,dx_1$, which is nonnegative\footnote{We must mention here that the assumption of the nonnegativity of $h(\vecc x) -f_2(1)f_1(x_1)$ could be replaced by that of $h(\vecc x)-f_2(1)f_1(x_1)$ being increasing with respect to $x_1$ and the argument would still carry through: we can move entire columns  of $S$ to the right, pushing elements horizontally; the projection towards axis $x_2$ again remains unchanged, and because of the monotonicity of $h(\vecc x)-f_2(1)f_1(x_1)$, the overall deficiency will not decrease since we are integrating over higher values of $x_1$.

This means that the monotonicity of $h(\vecc x)-f_j(x_j)j_{-j}(1)$ with respect to $x_j$ can replace its nonnegativity in the initial Assumption~\ref{assume:upwards_def} (while still maintaining the regularity requirement of $h(\vecc x)$ being nonnegative) without affecting the main results of this paper, namely Theorems \ref{th:characterization_main}, \ref{th:characterization_iid} and \ref{th:two_iid_approx}.
\label{foot:alter-assumption-monotone}
}.

So $S$ can be assumed to be the intersection of $D_{1,2}$ with a box, i.e.~$S=[t_1,1]\times[t_2,1]\inters D_{1,2}$, where $t_1\geq x_1^*$ and $t_2\geq x_2^*$. This also means that its projections are $S_1=[t_1,1]$ and $S_2=[t_2,1]$.
Now consider the lowest horizontal slice $\slice{S}{x_2}{t_2}$ of $S$. It obviously lies within $D_{1,2}$. But from condition~\eqref{eq:1slice_gen_integrals} so do all horizontal line segments of the form $[s_1(x_2),1]$ for any $x_2\in[x_2^*, t_2]$: $s_1(x_2)$ is decreasing and specifically less steeply than the line $-x_2+p$ which is the boundary of $D_{1,2}$. So, by adding all these segments to $S$ we won't increase the projections towards the $x_1$-axis (these are covered already by $\slice{S}{x_2}{t_2}$, which has to be a superset of $[s_1(t_2),1]$, otherwise it would have a negative deficiency, see~\eqref{eq:1slice_gen_integrals}) and the new projections towards the $x_2$-axis are dominated by the increase of the area of $S$ (this segments have nonnegative deficiency). So, $S$ can be assumed to project in the entire boundaries $[x_1^*,1]$ and $[x_2^*,1]$ of $D_{1,2}$ and thus, since $h$ is nonnegative, $S$ can be assumed to fill the entire $D_{1,2}$ region. But by the definition of price $p$ in Theorem~\ref{th:characterization_main}, $\delta(D_{1,2})=0$ which concludes the proof. 
\end{proof}
\subsection{Dual Solution and Optimality}
\label{sec:optimality}
Notice that Theorem~\ref{th:characterization_main} ensures the existence of a full-bundling price in~\eqref{eq:2slice_gen}. This needs to be proven. Indeed,
quantity $\int_Dh(\vecc x)\,d\vecc x$ continuously (weakly) increases as $p$ decreases, and for $p=0$ %
\begin{align*}
\int_Dh(\vecc x)\,d\vecc x &=\int_0^1\int_0^1 3f_1(x_1)f_2(x_2)+x_1f_1'(x_1)f_2(x_2)+yf_2'(x_2)f_1(x_1)\,dx_1\,dx_2\\
		&=3+(f_1(1)-1)+(f_2(1)-1)=1+f_1(1)+f_2(1)
		>f_1(1)+f_2(1)
\end{align*}
while for $p=\hat x_1+\hat x_2$, where $(\hat x_1,\hat x_2)$ is the unique point of intersection of the curves $x_2=s_1(x_1)$ and $x_1=s_2(x_2)$ in $I^2$ (such a point certainly exists because $s_1$ and $s_2$ are defined over the entire $I$), 
\begin{align*}
\int_{D_{1,2}}h(\vecc x)\,d\vecc x &=\int_{\hat x_2}^1\int_{\hat x_1}^1 h(\vecc x)\,d\vecc x
		\leq \int_{\hat x_2}^1\int_{s_1(x_2)}^1 h(\vecc x)\,d\vecc x
		= \int_{\hat x_2}^1f_1(1)f_2(x_2)\,dx_2\\
		&= f_1(1)(1-F_2(\hat x_2))
		\leq f_1(1)(1-F_2(\hat x_2))+f_2(1)(1-F_1(\hat x_1)),
\end{align*}
the first inequality holding because $h$ is nonnegative and $s_1(x_2)\leq s_1(\hat x_2)=\hat x_1$ ($s_1$ is decreasing), and the second equality by substituting~\eqref{eq:1slice_gen_integrals}, and from~\eqref{eq:2slice_gen_bundle_region} this means that $\int_{D}h\,d\vecc x\leq f_1(1)+f_2(1)$.

Combining the above, indeed there must be a $p\in[0,\hat x_1+\hat x_2]$ such that $\int_{D}h\,d\vecc x=f_1(1)+f_2(1)$. In fact, using this argument, if for $p=\min\sset{s_1(0),s_2(0)}$ it is $\int_{D}h\,d\vecc x<f_1(1)+f_2(1)$ then $p$ must go below this value to get a solution, meaning that the full-bundling region will cover the rest of the regions $D_1$ and $D_2$, i.e.~$D=D_{1,2}$, and the mechanism defined by~\eqref{eq:optimal_auction_gen} is a deterministic full-bundling.  

The following lemma will complete the proof of Theorem~\ref{th:characterization_main}. It is the most technical part of this paper, and utilizes a max-flow min-cut argument in order to prove the existence of a feasible dual pair $z_1,z_2$ that satisfies the complementarity conditions with respect to the utility function given by Theorem~\ref{th:characterization_main}, thus establishing optimality. It is inspired by the bipartite matching approach in~\citep{gk2014} where Hall's theorem is used in order to prove existence, in the special case of uniformly distributed items. Here we need to abstract and generalize our approach in order to incorporate general distributions in the most smooth way possible. The proof has a strong geometric flavor, which is achieved by utilizing the notion of deficiency that was introduced in Sect.~\ref{sec:nodef} and using Lemma~\ref{lemma:no_positive_def}.  
\begin{lemma}
\label{lemma:coloring}
Assume that the conditions of Theorem~\ref{th:characterization_main} hold. Then for arbitrary small $\varepsilon>0$, there exist feasible dual solutions $z_1,z_2$ which are $\varepsilon$-complementary to the (primal) $u$ given by~\eqref{eq:optimal_auction_gen}. Therefore, the mechanism induced by $u$ is optimal.
\end{lemma}
\begin{proof}
Following the discussion in Sect.~\ref{sec:duality}, we would like to decompose $W$ into the desired functions $w_1$ and $w_2$ within $D_{1,2}$, i.e.~such that they satisfy~\eqref{eq:Wdecomp_sum}--\eqref{eq:Wdecomp1}. In fact, we are aiming for $\varepsilon$-complementarity, so we can relax conditions~\eqref{eq:Wdecomp1} a bit: 
\begin{equation}
\int_0^1w_j(\vecc x)\,dx_j \leq f_j(1)f_{-j}(x_{-j})+\varepsilon' 
\label{eq:relax_dual_boundary}
\end{equation}
To be precise, the $\varepsilon$-complementarity of Lemma~\ref{lemma:complementarity} dictates that regarding these conditions we must show that for a.e.\ $\vecc x\in D_{1,2}$ property~\eqref{eq:e_compl_3} holds (conditions~\eqref{eq:e_compl_2} and~\eqref{eq:e_compl_4} are immediately satisfied with strong equality, by~\eqref{eq:prop_dual_1} and the fact that within $D_{1,2}$ both items are sold deterministically with probability $1$.).
But since $u(\vecc x)\leq x_1+x_2 \leq 2$ for all $x_1,x_2\in I$ ($u$'s derivatives are at most $1$ with respect to any direction) and also exists $M>0$ such that $f_1(1)f_2(x_2),f_2(1)f_1(x_1)\geq M$ for all $\vecc x\in D_{1,2}$ (the density functions are continuous over the closed interval $I$ and positive\footnote{We would like to note here that this is the only point in the paper where the fact that the densities are \emph{strictly} positive is used. As a matter of fact, a closer look will reveal that the proof just needs the property to hold in the closure of $D_{1,2}$ and not necessarily in the entire domain $I^2$. This allows the consideration of a wider family of feasible distributional priors, for example the monomial distributions of Corollary~\ref{th:optimal_two_power}: their densities $f(t)=(c+1)t^c$ may vanish at $t=0$ but these ``problematic'' points happen to lie outside the area $D_{1,2}$ where both items are sold.}), indeed~\eqref{eq:relax_dual_boundary} is enough to guarantee $\varepsilon$ complementarity if one ensures $\varepsilon'\leq \varepsilon M/2$. So, the remaining of the proof is dedicated into constructing nonnegative, a.e.~continuous functions $w_1$ and $w_2$ over $D_{1,2}$, such that $w_1+w_2=h$ and~\eqref{eq:relax_dual_boundary} are satisfied.

We will do that by constructing an appropriate graph and recovering $w_1$ and $w_2$ as ``flows'' through its nodes, deploying the min-cut max-flow theorem to prove existence. To start, we pick an arbitrary small $\delta>0$ and discretize $I^2$ into a lattice of $\delta$-size boxes $[(i-1)\delta,i\delta] \times [(j-1)\delta,j\delta]$, where $i,j=1,2,\dots,1/\delta$, selecting $\delta$ such that $1/\delta$ is an integer. Denote the intersection of such a box with $D_{1,2}$ by $B_{i,j}$. Also, let $B^1_i$ denote the projection of all nonempty $B_{i,j}$'s, as $j$ ranges, towards the $x_1$-axis and $B^2_j$ towards the $x_2$-axis, as $i$ ranges. Note that these are well-defined in this way, since by the geometry of region $D_{1,2}$ two nonempty $B_{i,j}$, $B_{i',j'}$ will have the same vertical projection if $i=i'$ and the same horizontal if $j=j'$. Also, it is a simple fact to observe that all $B^1_i$ and $B^2_j$ are single-dimensional real intervals of length at most $\delta$.

Now let's construct a directed graph $G=(V,E)$, together with a capacity function $c(e)$ for all edges $e\in E$. Initially, for any pair $(i,j)$ such that $B_{i,j}$ has positive (two-dimensional Lebesgue) measure we insert a node $v(i,j)$ in $V$. We'll call these nodes \emph{internal} and we'll denote them by $V_o$. Also, for any internal node $v(i,j)$ we add nodes $v_1(i)$ and $v_2(j)$ corresponding to entire columns and rows, calling them \emph{column} and \emph{row} vertices and denoting them by $V_1$ and $V_2$, respectively. Finally there are two special nodes, a source $\sigma$ and a destination $\tau$. From the source to all internal nodes $v=v(i,j)$ we add an edge $(\sigma,v)$ with capacity equal to the area of $B_{i,j}$ under $h$, i.e.~$c(\sigma,v)=\int_{B_{i,j}}h(\vecc x)\,d\vecc x$. From any internal node $v=v(i,j)$ to its external column and row nodes $v_1=v_1(i)$ and $v_2=v_2(j)$ we add edges with capacities $c(v,v_1)=c(v,v_2)=c(\sigma,v)$ equal to the internal node's incoming edge capacity from the source. Finally, for all external nodes $v_1(i)\in V_1$ and $v_2(j)\in V_2$ we add edges towards the destination $\tau$ with capacities $c(v_1,\tau)=f_2(1)\int_{B^1_i}f_1(x_1)\,dx_1$ and $c(v_2,\tau)=f_1(1)\int_{B^2_j}f_2(x_2)\,dx_2$, respectively. The structure of graph $G$ is depicted in Fig.~\ref{fig:flows_graph}.
 \begin{figure}
 \centering
 \includegraphics[width=10cm]{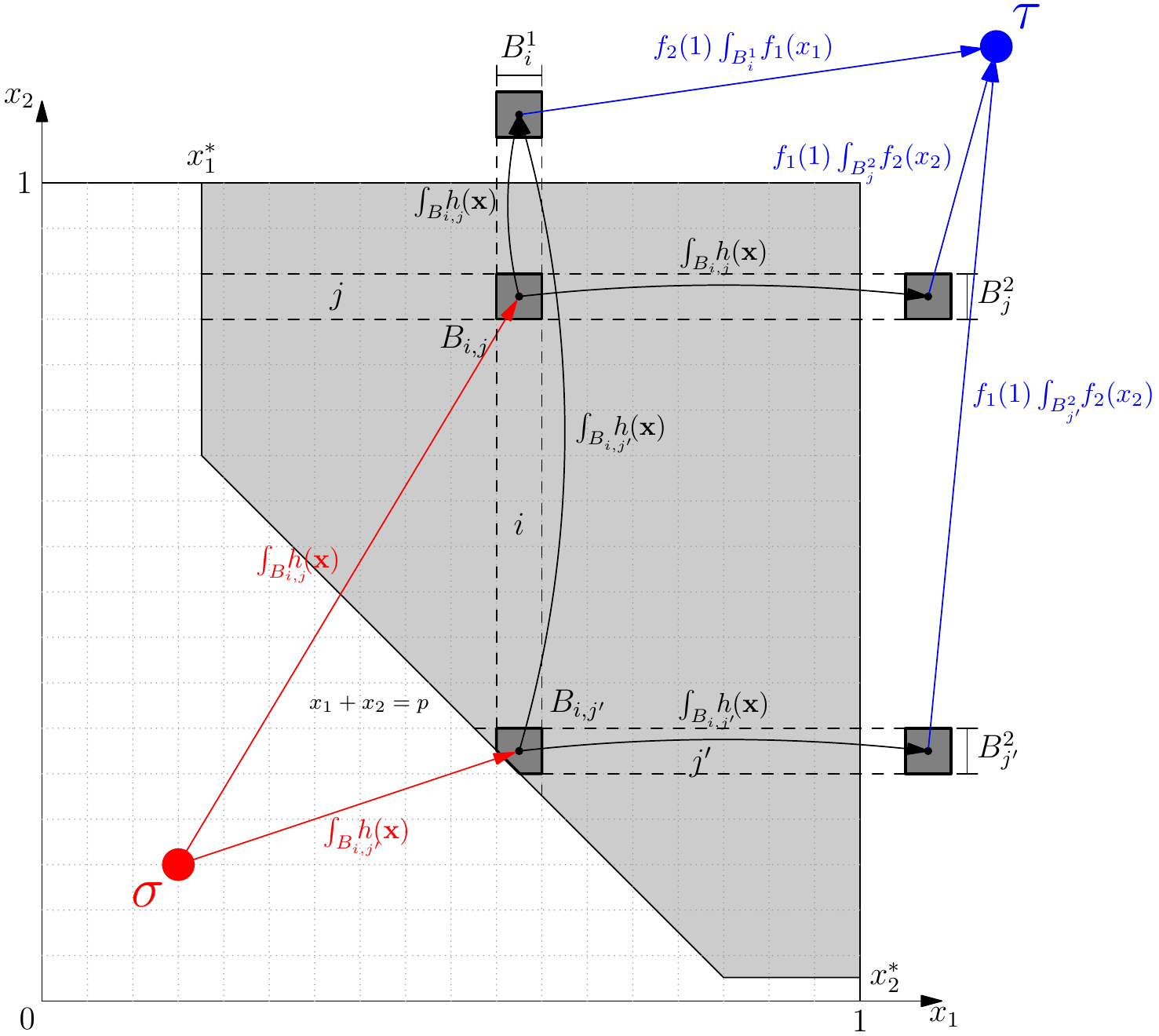}
 \caption{\footnotesize The graph $G$ in the proof of Lemma~\ref{lemma:coloring}. Every internal node $B_{i,j}$ of region $D_{1,2}$ can receive at most $\int_{B_{i,j}}h(\vecc x)\,d\vecc x$ flow from the source node $\sigma$ and can send at most that amount to each one of its neighbouring external nodes $B_{i}^1$ and $B_{j}^2$. Every external node $B^1_i$ and $B^2_j$ is connected to the destination $\tau$ with edges of capacity $f_2(1)\int_{B_i^1}f_1(x_1)\,dx_1$ and $f_1(1)\int_{B_j^2}f_2(x_2)\,dx_2$, respectively. Internal $B_{i,j}$'s are two-dimensional intersections of $\delta$-boxes with $D_{1,2}$, while the external ones, $B^1_i$ and $B^2_j$ are single dimensional intervals of length $\delta$.}
 \label{fig:flows_graph}
 \end{figure}

As a first observation, notice that the maximum flow that can be sent from $\sigma$ within the graph is $\int_{D_{1,2}}h(\vecc x)\,dx_1\,dx_2$ and the maximum flow that $\tau$ can receive is 
$$f_2(1)\int_{x_1^*}^1f_1(x_1)\,dx_1+f_1(1)\int_{x_2^*}^1f_2(x_2)\,dx_2$$ (remember that the projection of $D_{1,2}$ to the $x_1$-axis is $[x_1^*,1]$ and to the $x_2$-axis $[x_2^*,1]$). But, from the way the entire region $D$ is constructed, we know that the above two quantities are equal (see~\eqref{eq:2slice_gen_bundle_region}). Let's denote this value by $\psi$. Next, we will prove that indeed one can create a feasible flow through $G$ that achieves that maximum value $\psi$. From the max-flow min-cut theorem, it is enough to show that the minimum $(\sigma,\tau)$-cut of $G$ has a value of at least $\psi$. To do that, we'll show that $(\sigma,V\setminus\{\sigma\})$ is a minimum cut of $G$.

Indeed, let $(S,V\setminus S)$ be a $(\sigma,\tau)$-cut of $G$. First, let there be an edge $(v,v_j)$ crossing the cut, i.e.~$v\in S$ and $v_j\notin S$, with $v$ internal node and $v_j$ external. Then, by moving $v$ at the other side of the cut, i.e.~removing it from $S$, we would create at most a new edge contributing to the cut, namely $(\sigma,v)$ but also destroy at least one edge $(v,v_j)$. Since the capacities of these two edges are the same, the overall effect would be to get a new cut with weakly smaller value. So, from now on we can assume that for all edges $(v,v_j)$ of $G$, if $v\in S$ then also $v_j\in S$. Under this assumption, if $S_{o}=V_{o}\inters S$ denotes the set of internal nodes belonging at the left side of the cut, for every $v\in S_{o}$ all edges $(v,v_j)$ adjacent to $v$ will not cross the cut. However, this means that all edges $(v_j,\tau)$, where $v_j\in N(v)$\footnote{$N(v)$ denotes the set of neighbours of $v$ in graph $G$.}, do contribute to the cut. But then, if we remove all nodes in $S_{o}$, together with their neighbouring external nodes $N(S_{o})$ at the other side of the cut, we increase the cut's value by at most $\sum_{v\in S_{o}}c(\sigma,v)$ and at the same time reduce it by at least $\sum_{v_j\in N(S_{o})}c(v_j,\tau)$. However, by the way graph $G$ is constructed, this corresponds to an overall increase in the cut of at least 
$$\int_B h(\vecc x)\,d\vecc x -f_2(1) \int_{B_{1}}f_1(x_1)\,dx_1 -f_1(1)\int_{B_{2}}f_2(x_2)\,dx_2,$$ 
where $B=\union_{v(i,j)\in S_{o}}B_{i,j}$ is the region of $D_{1,2}$ covered by the boxes of nodes in $S_{o}$ and $B_1$, $B_2$ are the projections of this body to the horizontal and vertical axis, respectively. From Lemma~\ref{lemma:no_positive_def} this difference must be nonpositive, thus this change results in a cut of an even (weakly) smaller value. The above arguments show that indeed the cut that has only $\sigma$ remaining at its left side is a minimum one.

So, there must be a flow $\phi:E\map\R_+$, achieving to transfer a total value of $\psi$ through $G$. As we argued above though, by the construction of $G$, in order to achieve this value of $\psi$ the full capacity of \emph{all} edges $(\sigma, v)$ as well as that of all $(v_j,\tau)$ must be used. So, this flow $f$ manages to elegantly separate all incoming flow $\phi(\sigma,v(i,j))=\int_{B_{i,j}}h(\vecc x)\,d\vecc x$ towards an internal box of $D_{1,2}$, into a sum of flows $\phi(v(i,j),v_1(i))+\phi(v(i,j),v_2(j))$ towards its external neighbours. But this is exactly what we need in order to construct our feasible dual solution! For simplicity, denote this incoming flow $\phi(i,j)$ and the outgoing ones $\phi_1(i,j)$ and $\phi_2(i,j)$, respectively. Then, define the functions $w_1$, $w_2$ throughout $D_{1,2}$ by
$$
w_1(\vecc x)=\frac{\phi_1(i,j)}{\phi(i,j)}h(\vecc x)\qquad\text{and}\qquad w_2(\vecc x)=\frac{\phi_2(i,j)}{\phi(i,j)}h(\vecc x),
$$
where $B_{i,j}$ is the discretization box where point $\vecc x$ of $D_{1,2}$ belongs to. In that way, first notice that we achieve $w_1+w_2=h$. Secondly, functions $w_1$ and $w_2$ are almost everywhere continuous, since the values of the flows are constant within the boxes, and our discretization is finite. The only remaining property to prove is~\eqref{eq:relax_dual_boundary}. 

Fix some height $x_2=\tilde x_2$ such that this horizontal line intersects $D_{1,2}$. We'll prove that 
$$\int_{0}^{1}w_1(x_1,\tilde x_2)\,dx_1-f_1(1)f_2(\tilde x_2)\leq \varepsilon'.$$ 
Value $\tilde x_2$ falls within some interval of the discretization, let $\tilde x_2\in[(\tilde j-1)\delta,\tilde j\delta]=B_{\tilde j}^2$. The average value of function $f_1(1)f_2(x_2)$ (with respect to $x_2$) within this interval is $$\frac{1}{\delta}f_1(1)\int_{B^2_{\tilde j}}f_2(x_2)\,dx_2=c(v_2(\tilde j),\tau)/\delta$$ and the average value of $\int_0^1w_1(\vecc x)\,dx_1$ is 
$$
\frac{1}{\delta}\int_{B^2_{\tilde j}}\int_0^1w_1(\vecc x)\,dx_1 = \frac{1}{\delta}\sum_i\int_{B_{i,\tilde j}}w_1(\vecc x)\,d\vecc x  = \frac{1}{\delta}\sum_i\frac{\phi_1(i,j)}{\phi(i,j)}\int_{B_{i,\tilde j}}h(\vecc x)\,d\vecc x 
= \sum_i \phi_1(i,\tilde j)/\delta.
$$
But since the sum of the outgoing flows over any horizontal line of internal nodes of the graph (here $j=\tilde j$) must equal the outgoing flow of the corresponding external node (here $v_2(\tilde j)$), the above quantities are equal. Thus, by selecting the discretization parameter $\delta$ small enough, we can indeed make the values $\int_{0}^{1}w_1(x_1,\tilde x_2)\,dx_1$ and $f_1(1)f_2(\tilde x_2)$ to be $\varepsilon'$ close to each other
\footnote{This should feel intuitively clear, and it relies on the uniform continuity of functions $f_2$ and $h$, but we also give a formal proof in Appendix~\ref{append:technical_flow_rest}.}.
\end{proof}

\section{The Case of Identical Items}
In this section we focus on the case of identically distributed values, i.e.~$f_1(t)=f_2(t)\equiv f(t)$ for all $t\in I$, and we provide clear and simple conditions under which the critical property~\eqref{eq:1slice_gen_functions} of Theorem~\ref{th:characterization_main} holds.  

First notice that in this case the regularity Assumption~\ref{assume:regularity} gives $3+\frac{x_1f'(x_1)}{f(x_1)}+\frac{x_2f'(x_2)}{f(x_2)}\geq 0$ a.e. in $I^2$ (since $f$ is positive) and thus $\frac{tf'(t)}{f(t)}\geq -\frac{3}{2}$ for a.e. $t\in I$. An equivalent way of writing this is that $t^{3/2}f(t)$ is increasing, which interestingly is the complementary case of that studied by~\citet{Hart:2012uq} for two i.i.d. items: they show that when $t^{3/2}f(t)$ is decreasing, then deterministically selling in a full bundle is optimal.

\begin{theorem}
\label{th:characterization_iid}
Assume that $G(t)=tf(t)/(1-F(t))$ and $H(t)=tf'(t)/f(t)$ give rise to well defined, differentiable functions over $I$, $G$ being strictly increasing and convex, $H$ decreasing and concave, with $G+H$ increasing and $G(1)\geq 2+H(0)$. Then the requirements of Theorem~\ref{th:characterization_main} are satisfied. In particular 
$$s(t)=G^{-1}(2+H(t))$$
and, if  
\begin{equation}
\label{eq:two_iid_full_bundle_price}
\int_0^1\int_0^1h(\vecc x)\,d\vecc x-\int_0^p\int_0^{p-x_2}h(\vecc x)\,d\vecc x-2f(1)
\end{equation}
is nonpositive for $p=s(0)$ then the optimal selling mechanism is the one offering deterministically the full bundle for a price of $p$ being the root of \eqref{eq:two_iid_full_bundle_price} in $[0,s(0)]$, otherwise the optimal mechanism is the one defined by the utility function
$$
u(\vecc x)=\max\sset{0,x_1-s(x_2),x_2-s(x_1),x_1+x_2-p}
$$
with $p=x^*+s(x^*)$, where $x^*\in [0,s(0)]$ is the constant we get by solving 
\begin{equation}
\label{eq:price_bundle_eq_iid}
\int_{x^*}^{s(x^*)}\int_{s(x^*)+x^*-x_2}^1h(\vecc x)\,d\vecc x+\int_{s(x^*)}^1\int_{x^*}^1h(\vecc x)\,d\vecc x= 2f(1)(1-F(x^*)).
\end{equation}
\end{theorem}
\begin{proof}
Function $G$ is strictly monotone, thus invertible and has a range of $[G(0),G(1)]=[0,G(1)]\supseteq [0,2+H(0)]$. By Assumption~\ref{assume:regularity} and the previous discussion, it must be $tf'(t)/f(t)\geq -3/2$, so $2+H(t)\geq 1/2>0$ for all $t\in I$. Thus, $s(t)=G^{-1}(2+H(t))$ is well defined and furthermore it is decreasing, since $G$ is increasing and $H$ decreasing. Also, by the way $s$ is defined we get that for all $t$: $G(s(t))=2+H(t)$, which is exactly condition~\eqref{eq:1slice_gen_functions} of Theorem~\ref{th:characterization_main}. 

It remains to be shown that $s$ is concave and that $s'(t)>-1$. From the definition of $s$, $s'(t)=H'(t)/G'(s(t))$. Function $H$ is decreasing and concave, so $H'(t)$ is negative and decreasing, and function $G$ is increasing and convex and $s$ decreasing, so $G'(s(t))$ is positive and decreasing. Combining these we get that the ratio $H'(t)/G'(s(t))$ is decreasing, proving that $s$ is concave. Finally, notice that since we are in a two item i.i.d. setting, the only part of curve $x_2=s(x_1)$ that matters and may appear in the utility of the resulting mechanism~\eqref{eq:optimal_auction_gen} is the one where $x_1\leq x_2$ (curves $x_2=s(x_1)$ and $x_1=s(x_2)$ will intersect on the line $x_1=x_2$), so we only have to show that $s'(t)>-1$ for $t\leq s(t)$. Indeed, in that case $G'(t)\leq G'(s(t))$, so $s'(t)=H'(t)/G'(s(t))\geq H'(t)/G'(t)$  and thus it is enough to show that $H'(t)-G'(t)\geq 0$ which we know holds since $H+G$ is assumed to be increasing.

\end{proof}

\begin{corollary}[Monomial Distributions]
\label{th:optimal_two_power}
The optimal selling mechanism for two items with i.i.d.\ values from the family of distributions with densities $f(t)=(c+1)t^c$, $c\geq 0$, is deterministic. In particular, it offers each item for a price of $s=\sqrt[c+1]{\frac{c+2}{2c+3}}$ and the full bundle for a price of $p=s+x^*$, where $x^*$ is the solution to~\eqref{eq:price_bundle_eq_iid}.
\end{corollary}
\begin{proof}
For two monomial i.i.d.\ items with $f_1(t)=f_2(t)=(c+1)t^c$ we have $h(\vecc x)=(c+1)^2 (2 c+3) x_1^c x_2^c\geq 0$, thus  $h(\vecc x)-f_2(1)f_1(x_1)=(c+1)^2 x_1^c \left((2 c+3) x_2^c-1\right)$ which is nonnegative for all $x_2\geq \sqrt[c]{1/(2c+3)}\equiv \omega$. So, in order to make sure that Assumption~\ref{assume:upwards_def} is satisfied, it is enough to show that $x^*\geq\omega$ because then $D_{1,2}\subseteq [\omega,1]^2$. We'll soon show that this is indeed satisfied for all $c\geq 0$.

Applying Theorem~\ref{th:characterization_iid} we compute: $G(t)=(c+1)t^{c+1}/(1-t^{c+1})$ which is strictly increasing and convex in $I$ and $H(t)=c$ which is constant and thus decreasing and concave. Also, it is trivial to deduce that $G+H$ is increasing and $\lim_{t\to 1^{-}}G(t)=\infty>2+c=2+H(0)$. Then, it is valid to compute $G^{-1}(t)=\left(\frac{3+2c}{2+c}\right)^{-\frac{1}{1+c}}$ and thus $s(t)=\sqrt[c+1]{\frac{c+2}{2c+3}}$ which is constant. 

Regarding the computation of the full-bundle price $p$, condition~\eqref{eq:price_bundle_eq_iid} gives rise to quantity
$$
\int_{x^*}^{s}\int_{s+x^*-x_2}^1x_1^cx_2^c\,d\vecc x+\int_{s}^1\int_{x^*}^1x_1^cx_2^c\,d\vecc x- \frac{2}{c+1}(1-{x^*}^{c+1}),
$$
which by plugging-in $x^*=\omega$ and using the values of $s$ and $\omega$ (as functions of $c$) one can see that  it is positive for all $c\geq 0$. So, by the discussion in the beginning of Sect.~\ref{sec:optimality} it can be deduced that the solution to~\eqref{eq:price_bundle_eq_iid}  will be such that $x^*>\omega$. 
\end{proof}
Notice that for $c=0$ the setting of Corollary~\ref{th:optimal_two_power} reduces to a two uniformly distributed goods setting, and gives the well-known results of $s=2/3$ and $p=(4-\sqrt{2})/3$ (see e.g.~\citep{Manelli:2006vn}). For the linear distribution $f(t)=2t$, where $c=1$, we get $s=\sqrt{3/5}$ and $p\approx1.091$.

\begin{corollary}[Exponential Distributions]
\label{th:optimal_two_expo}
The optimal selling mechanism for two items with exponentially i.i.d.\ values over $[0,1]$, i.e.\ having densities $f(t)=\lambda e^{-\lambda t}/(1-e^{-\lambda})$, with $0<\lambda\leq 1$, is the one having $s(t)=\frac{1}{\lambda}\left[2-\lambda t-W\left(e^{2-\lambda-\lambda t}(2-\lambda t)\right)\right]$ and a price of
$p=x^*+s(x^*)$ for the full bundle, where $x^*$ is the solution to~\eqref{eq:price_bundle_eq_iid}. Here $W$ is Lambert's product logarithm function\footnote{Function $W$ can be defined as the solution to $W(t)e^{W(t)}=t$.}.
\end{corollary}
\begin{proof}
For two i.i.d. exponentially distributed items with $f_1(t)=f_2(t)=\lambda e^{-\lambda t}/(1-e^{-\lambda})$ we have 
$$
\hspace{-0.5cm}
h(\vecc x)-f_2(1)f_1(x_1)=\frac{\lambda^2}{\left(e^\lambda-1\right)^2} e^{2-\lambda (x_1+x_2)} (3-\lambda (x_1+x_2)-e^{\lambda x_2})\geq \frac{\lambda^2}{\left(e^\lambda-1\right)^2} e^{2-\lambda (x_1+x_2)} (2-\lambda (x_1+x_2))\geq 0
$$ 
for all $x_1,x_2\in I$, since $\lambda\leq 1$.

Applying Theorem~\ref{th:characterization_iid} we compute: $G(t)=\lambda t/(1-e^{-\lambda(1-t)})$ which is strictly increasing and convex in $I$ and $H(t)=-\lambda t$ which is decreasing and concave. Also, $G(t)+H(t)=\lambda te^{-\lambda(1-t)}/(1-e^{-\lambda(1-t)})$ is increasing and $\lim_{t\to 1^{-}}G(t)=\infty>2=2+H(0)$. Then, it is valid to compute $G^{-1}(t)=t/\lambda-W\left(t e^{t-\lambda}\right)/\lambda$ and thus $s(t)=\frac{1}{\lambda}\left[2-\lambda t-W\left(e^{2-\lambda-\lambda t}(2-\lambda t)\right)\right]$.
\end{proof}

For example, for $\lambda=1$ we get $s(t)=2-t-W\left(e^{1-t} (2-t)\right)$ and $p\approx 0.714$. 
Interestingly, to our knowledge this is the first example for an i.i.d.\ setting with values coming from a regular, continuous distribution over an interval $[0,b]$, where an optimal selling mechanism is \emph{not} deterministic. Also notice how this case of exponential i.i.d.\ items on a bounded interval is different from the one on $[0,\infty)$: by \citep{Daskalakis:2013vn,g2014} we know that at the unbounded case the optimal selling mechanism for two exponential i.i.d.\ items is simply the deterministic full-bundling, but in our case of the bounded $I$ this is not the case any more.

\section{Non-Identical Items}
\label{sec:non-iid}

An interesting aspect of the technique of Theorem~\ref{th:characterization_iid} is that it can readily be used also for non identically distributed values. One just has to define $G_j(t)\equiv tf_j(t)/(1-F_j(t))$ and $H_j(t)=tf_j'(t)/f_j(t)$ for both items $j=1,2$ and check again whether $G_1,G_2$ are strictly increasing and convex and $H_1,H_2$ nonnegative, decreasing and concave. Then, we can get $s_j(t)=G_j^{-1}(2+H_{-j}(t))$ and check if $s_j(1)> -1$ and the price $p$ of the full bundle can be given by~\eqref{eq:2slice_gen}. Again, a quick check of whether full bundling is optimal is to see if for $p=\min\sset{s_1(0),s_2(0)}$ expression $\int_0^1\int_0^1h(\vecc x)\,d\vecc x-\int_0^p\int_0^{p-x_2}h(\vecc x)\,d\vecc x-f_1(1)-f_2(1)$ is nonpositive.

\begin{example}
\label{example:uniform-expo}
Consider two independent items, one having uniform valuation $f_1(t)=1$ and one exponential $f_2(t)=e^{-t}/(1-e^{-1})$. Then we get that $s_1(t)=(2-t)/(3-t)$, $s_2(t)=2-W(2e)\approx 0.625$ and $p \approx0.787$. The optimal selling mechanism offers either only item $2$ for a price of $s_2\approx 0.625$, or item $1$ deterministically and item $2$ with a probability $s_1'(x_2)$ for a price of $s_1(x_2)-x_2s_1'(x_2)$, or the full bundle for a price of $p\approx 0.787$. You can see the allocation space of this mechanism in Fig.~\ref{fig:Exp_Uniform}.
\end{example}

\section{Approximate Solutions}
\label{sec:approximate_convex_fail}
In the previous sections we developed tools that, under certain assumptions, can give a complete closed-form description of the optimal selling mechanism. However, remember that the initial primal-dual formulation upon which our analysis was based, assumes a relaxed optimization problem. Namely, we dropped the convexity assumption of the utility function $u$. In the results of the previous sections this comes for free: the optimal solution to the relaxed program turns out to be convex anyways, as a result of the requirements of Theorem~\ref{th:characterization_main}. But what happens if that was not the case? The following tool shows that even in that case our results are still applicable and very useful for both finding good upper bounds on the optimal revenue (Theorem~\ref{th:two_iid_approx}) as well as designing almost-optimal mechanisms that have provably very good performance guarantees (Sect.~\ref{sec:convexification}).
\begin{theorem}
\label{th:two_iid_approx}
Assume that all conditions of Theorem~\ref{th:characterization_main} are satisfied, except from the concavity of functions $s_1,s_2$. Then, the function $u$ given by that theorem might not be convex any more and thus not a valid utility function, but it generates an \emph{upper bound} to the optimal revenue, i.e.\ $\rev(f_1,f_2)\leq\mathcal R_{f_1,f_2}(u)$. In particular, this is the case if all the requirements of Theorem~\ref{th:characterization_iid} hold except the concavity of $H$.
\end{theorem}
\begin{proof}
The proof is a straightforward result of the duality framework (see Sect.~\ref{sec:duality}): By dropping only the concavity requirement of functions $s_1$ and $s_2$ but satisfying all the remaining conditions of Theorem~\ref{th:characterization_main}, we still construct an optimal solution to the pair of primal-dual programs, meaning that function $u$ produced in \eqref{eq:optimal_auction_gen} maximizes $\mathcal R_{f_1,f_2}(u)$ over the space of all functions $u: I^2\map \R_+$ with partial derivatives in $[0,1]$ (see~\eqref{eq:allocs_probs_01}); the only difference is that $u$ might not be convex since $s_1,s_2$ might not be concave any more. The actual optimal revenue objective $\rev(f_1,f_2)$ has the extra constraint of $u$ being convex, thus, given that it is a maximization problem, it has to be that $\rev(f_1,f_2)\leq \mathcal R_{f_1,f_2}(u)$. Finally, it is easy to verify in the proof of Theorem~\ref{th:characterization_iid} that dropping just the concavity requirement for $H$ can only affect the concavity of functions $s_1,s_2$ and hence the convexity of $u$.
\end{proof}

\begin{example}[Power-Law Distributions]
\label{example:power-law}
 A class of important distributions that falls into the description of Theorem~\ref{th:two_iid_approx} are the power-law distributions with parameters $0<\alpha\leq 2$. More specifically, these are the distributions having densities $f(t)=c/(t+1)^\alpha$, with the normalization factor $c$ selected so that $\int_0^1f(t)\,dt=1$, i.e.~$c=(a-1)/(1-2^{1-\alpha})$. It is not difficult to verify that these distributions satisfy Assumption~\ref{assume:upwards_def}. For example, for $\alpha=2$ one gets $f(x)=2/(x+1)^2$, the \emph{equal revenue} distribution shifted in the unit interval. For this we can compute via~\eqref{th:two_iid_approx} that $s(t)=\frac{1}{2} \sqrt{5+2 t+t^2}-\frac{1}{2} (1+t)$ and $p\approx 0.665$, which gives an upper bound of $R_{f,f}(u)\approx 0.383$ to the optimal revenue $\rev(f,f)$.
\end{example}

\subsection{Convexification}
\label{sec:convexification}
The approximation results described in Theorem~\ref{th:two_iid_approx} can be used not only for giving upper bounds on the optimal revenue, but also as a \emph{design} technique for good selling mechanisms. Since the only deviation from a feasible utility function is the fact that function $s$ is not concave (and thus $u$ is not convex), why don't we try to ``convexify'' $u$, by replacing $s$ by a concave function $\tilde s$? If $\tilde s$ is ``close enough'' to the original $s$, by the previous discussion this would also result in good approximation ratios for the new, feasible selling mechanism. 

Let's demonstrate this by an example, using the equal revenue distribution $f(t)=2/(t+1)^2$ of the previous example. We need to replace $s$ with a concave $\tilde s$ in the interval $[0,x^*]$. So let's choose $\tilde s$ to be the concave hull of $s$, i.e.~the minimum concave function that dominates $s$. Since $s$ is convex, this is simply the line that connects the two ends of the graph of $s$ in $[0,x^*]$, that is, the line
$$
\tilde s(t)=\frac{s(0)-s(x^*)}{x^*}(x^*-t)+s(x^*).
$$
A calculation shows that this new \emph{valid} mechanism has an expected revenue which is within a factor of just $1+3\times 10^{-9}$ of the upper bound given by $s$ using Theorem~\ref{th:two_iid_approx}, rendering it essentially optimal.

\paragraph{Acknowledgements:} We thank Anna Karlin, Amos Fiat, Costis Daskalakis and Ian Kash for insightful discussions. We also thank the anonymous reviewers for their useful comments on the conference version of this paper.

\bibliographystyle{abbrvnat} 
\bibliography{TwoItems}

\appendix

\section{Remaining Proof of Lemma~\ref{lemma:coloring}}
\label{append:technical_flow_rest}
Functions $f_2$ and $\int_0^1w_1(x_1,\tilde x_2)\,dx_1$ are continuous in the interval $B^2_{\tilde j}$, so by the Mean Value Theorem there exist $\bar x_2,\bar{\bar x}_2\in B^2_{\tilde j}$ such that
\begin{equation}
\label{eq:mean_value_theorem_flow}
\int_0^1w_1(x_1,\bar x_2)\,dx_1=\frac{1}{\delta}\int_{B^2_{\tilde j}}\int_0^1w_1(\vecc x)\,d\vecc x=\frac{1}{\delta}f_1(1)\int_{B^2_{\tilde j}}f_2(x_2)\,dx_2=f_1(1)f_2(\bar{\bar x}_2)
\end{equation}
Notice that both $\bar x_2$ and $\bar{\bar x}_2$ are $\delta$-close to $\tilde x_2$. Function $f_2$ is uniformly continuous, so one can pick $\delta$ small enough in order to 
\begin{equation}
\label{eq:uniform_cont_1}
f_1(1)f_2(\bar{\bar x}_2)-f_1(1)f_2(\tilde x_2)\leq \varepsilon'/2.
\end{equation}
In the same way, because $h$ is uniformly continuous, we can select $\delta$ small enough so that $h(x_1,\bar x_2)-h(x_1,\tilde x_2)\leq \varepsilon'/3$ for all $x_1\in I$, and that would give
\begin{align}
\left | \int_0^1w_1(x_1,\bar x_2)\,dx_1-\int_0^1w_1(x_1,\tilde x_2)\,dx_1\right | &\leq \sum_i\frac{f_1(i,j)}{f(i,j)}\int_{B^1_i}\left | h(x_1,\bar x_2)- h(x_1,\tilde x_2)\right |\,dx_1 \notag\\
 &+\left | \bar x_2-\tilde x_2\right |\left\| h\right\|_\infty \notag \\
&\leq \sum_i\int_{B^1_i}\left | h(x_1,\bar x_2)- h(x_1,\tilde x_2)\right |\,dx_1 +\delta\left\| h\right\|_\infty \notag \\
&\leq \int_0^1\frac{\varepsilon'}{3}\,dx_1 +\delta\left\| h\right\|_\infty \notag \\
&\leq \varepsilon'/2, \label{eq:uniform_cont_2}
\end{align}
for choosing a small enough value for $\delta$, since $\left\| h\right\|_\infty\equiv\sup_{\vecc x\in I^2} h(\vecc x)$ is a fixed constant (because $h$ is continuous). The last additive term in the first inequality accounts for the fact that the length of the intersections of horizontal lines $x_2=\bar x_2$ and $x_2=\tilde x_2$ with $D_{1,2}$ may differ by $\left | \bar x_2 - \tilde x_2\right|$ (remember that the boundary of $D_{1,2}$ is a $45^\circ$--line).

Finally, by plugging in inequalities~\eqref{eq:uniform_cont_1} and~\eqref{eq:uniform_cont_2} into~\eqref{eq:mean_value_theorem_flow} we get the desired
$$
\left | \int_0^1w_1(x_1,\tilde x_2)\,dx_1 -f_1(1)f_2(\tilde x_2)\right |\leq \varepsilon '.
$$
\end{document}